\documentclass[envcountsame]{llncs}

\usepackage{hyperref}
\usepackage{colortbl}
\usepackage{xspace}
\usepackage{amsmath}
\usepackage{amssymb}
\usepackage{ifthen}
\usepackage{xspace}
\usepackage[dvipsnames]{xcolor}
\usepackage{subcaption}
\usepackage{dcolumn}
\usepackage{tikz}
\usepackage{booktabs}
\usepackage{graphicx}
\usepackage[linesnumbered,ruled,noend,vlined]{algorithm2e}
\usepackage{paralist}
\usepackage[capitalise]{cleveref}
\usepackage{environ}
\usepackage{wrapfig}
\usepackage{thm-restate}
\usepackage{xcolor}
\usepackage{tcolorbox}
\usepackage{typed-checklist}

\usepackage{fontawesome}

\NewEnviron{killcontents}{}

\usetikzlibrary{automata,positioning,trees,shapes,petri,arrows,snakes,backgrounds,calc,fit,arrows.meta,shadows,math,decorations.markings}
\usetikzlibrary{decorations.pathmorphing}
\usetikzlibrary{shapes.multipart,arrows,automata}
\usetikzlibrary{shapes.geometric}
\usetikzlibrary{positioning}
\usetikzlibrary{calc,patterns,patterns.meta}

\def\blindreview{}

\newif\ifTR
\ifTR
\else
\fi

\newcounter{claimcounter}

\crefname{claimcounter}{Claim}{Claims}

\Crefname{algocf}{Algorithm}{Algorithms}
\crefname{figure}{Fig.}{Figs.}
\crefname{section}{Sec.}{Secs.}

\InputIfFileExists{cutmargins.tex}{%
	\pdfpagesattr{/CropBox [92 112 523 758]} 
}{%
}

\newcommand{\vh}[1]{\textcolor{orange}{\ifmmode \text{[#1]}\else [VH: #1] \fi}}
\newcommand{\ol}[1]{\textcolor{blue}{\ifmmode \text{[OL: #1]}\else [OL: #1] \fi}}
\newcommand{\bs}[1]{\textcolor{ForestGreen}{\ifmmode \text{[BS: #1]}\else [BS: #1] \fi}}

\newcommand{\blinded}[1]{\ifx\blindreview\undefined #1 \else \textcolor{black!65}{[blinded for review]}\fi}

\newcommand{\infofst}[1]{\inf_{Q,\delta}(#1)}

\newcommand{\restrof}[2]{#1 \raisebox{-.5ex}{$|$}_{#2}}

\newcommand{\aut}[0]{\mathcal{A}}

\newcommand{\autw}[0]{\mathcal{W}}
\newcommand{\autc}[0]{\mathcal{C}}

\newcommand{\pr}[0]{\mathit{pr}}
\newcommand{\sat}[0]{\mathit{sat}}

\newcommand{\transover}[1]{\overset{#1}{\rightarrow}}
\newcommand{\ltr}[1]{\transover{#1}}

\newcommand{\word}[0]{\alpha}
\newcommand{\wordof}[1]{\seqof{\word}{#1}}

\newcommand{\seqof}[2]{#1_{#2}}

\newcommand{\algmaxrank}[0]{\textsc{MaxRank}\xspace}

\newcommand{\lang}[0]{\mathcal{L}}
\newcommand{\langof}[1]{\lang(#1)}
\newcommand{\langautof}[2]{\lang_{#1}(#2)}

\newcommand{\accstates}{Q_F}
\newcommand{\acctrans}{\delta_F}
\newcommand{\trans}{\delta}

\newcommand{\simby}[0]{\preceq}

\newcommand{\fairsimby}[0]{\mathrel{\simby_{f}}}

\newcommand{\fairsimbyc}[0]{\mathrel{\simby_{f}^{\autc}}}

\newcommand{\dirsimbyw}[0]{\mathrel{\simby_{\mathit{di}}^{\autw}}}

\newcommand{\claimqed}[0]{\hfill $\blacksquare$}

\newcommand{\ncsbmaxrank}[0]{\text{NCSB-}\algmaxrank}
\newcommand{\ncsblazy}[0]{\text{NCSB-}\textsc{Lazy}\xspace}

\newcommand{\ranker}[0]{\textsc{Ranker}\xspace}

\newcommand{\rankerold}[0]{\textsc{Ranker}_{\textsc{Old}}\xspace}

\newcommand{\spot}[0]{\textsc{Spot}\xspace}
\newcommand{\seminator}[0]{\textsc{Seminator}~2\xspace}
\newcommand{\goal}[0]{\textsc{Goal}\xspace}
\newcommand{\roll}[0]{\textsc{Roll}\xspace}
\newcommand{\fribourg}[0]{\textsc{Fribourg}\xspace}
\newcommand{\piterman}[0]{\textsc{Piterman}\xspace}
\newcommand{\safra}[0]{\textsc{Safra}\xspace}
\newcommand{\autfilt}[0]{\texttt{autfilt}\xspace}
\newcommand{\ltldstar}[0]{\textsc{LTL2dstar}\xspace}
\newcommand{\goalmark}[0]{\scriptsize \faSoccerBallO}
\newcommand{\automizer}[0]{\textsc{Ultimate Automizer}\xspace}

\definecolor{rowgray}{gray}{0.85}

\makeatletter
\newcommand{\monus}{\mathbin{\text{\@dotminus}}}
\newcommand{\@dotminus}{%
  \ooalign{\hidewidth\raise1ex\hbox{.}\hidewidth\cr$\m@th-$\cr}%
}
\makeatother

\newcommand{\dsrandom}[0]{\underline{\texttt{random}}\xspace}
\newcommand{\dsltl}[0]{\underline{\texttt{LTL}}\xspace}
\newcommand{\dsall}[0]{\underline{\texttt{all}}\xspace}
\newcommand{\dsboth}[0]{\underline{\texttt{both}}\xspace}
\newcommand{\dsautomizer}[0]{\underline{\texttt{Automizer}}\xspace}

\newcommand{\mihay}[0]{\textsc{MiHay}\xspace}
\newcommand{\cobapar}[0]{\mihay_\theta}
\newcommand{\cobapr}[0]{\mihay_\pr}
\newcommand{\cobasat}[0]{\mihay_\sat}

\newcommand{\delnondet}{{\color{blue}{\delta_1}}}
\newcommand{\deltrans}{{\color{red}{\delta_t}}}
\newcommand{\deldet}{{\color{green!60!black}{\delta_2}}}
\newcommand{\stnondet}{{\color{blue}Q_1}}
\newcommand{\stdet}{{\color{green!60!black}Q_2}}

\newcommand{\precircle}{node[left,shape=circle,inner sep=1pt,scale=0.7,fill=orange!60!black,text=white,draw=orange!60!black,xshift=-1mm,yshift=1mm]  {$1$}}
\newcommand{\textprecircle}{ (\tikz[baseline,anchor=base,scale=0.5]{ \draw \precircle;})}

\newcommand{\featcircle}{node[left,shape=circle,inner sep=1pt,scale=0.7,fill=orange!60!black,text=white,draw=orange!60!black,xshift=-1mm,yshift=1mm]  {$2$}}
\newcommand{\textfeatcircle}{ (\tikz[baseline,anchor=base,scale=0.5]{ \draw \featcircle;})}

\newcommand{\postcircle}{node[left,shape=circle,inner sep=1pt,scale=0.7,fill=orange!60!black,text=white,draw=orange!60!black,xshift=-1mm,yshift=1mm]  {$3$}}
\newcommand{\textpostcircle}{ (\tikz[baseline,anchor=base,scale=0.5]{ \draw \postcircle;})}

\newcommand{\iwacircle}{node[left,shape=circle,inner sep=1pt,scale=0.7,fill=orange!60!black,text=white,draw=orange!60!black,xshift=-1mm,yshift=1mm]  {A}}
\newcommand{\textiwacircle}{ (\tikz[baseline,anchor=base,scale=0.5]{ \draw \iwacircle;})}

\newcommand{\sdcircle}{node[left,shape=circle,inner sep=1pt,scale=0.7,fill=orange!60!black,text=white,draw=orange!60!black,xshift=-1mm,yshift=1mm]  {B}}
\newcommand{\textsdcircle}{ (\tikz[baseline,anchor=base,scale=0.5]{ \draw \sdcircle;})}

\newcommand{\gencircle}{node[left,shape=circle,inner sep=1pt,scale=0.7,fill=orange!60!black,text=white,draw=orange!60!black,xshift=-1mm,yshift=1mm]  {C}}
\newcommand{\textgencircle}{ (\tikz[baseline,anchor=base,scale=0.5]{ \draw \gencircle;})}

\newif\ifTR
\TRtrue
\title{Complementing B\"{u}chi Automata with \ranker (Technical Report)}

\author{
  Vojt\v{e}ch Havlena \and
  Ond\v{r}ej Leng\'{a}l \and
	Barbora \v{S}mahl\'{i}kov\'{a}
  }

\institute{
  Faculty of Information Technology,
  Brno University of Technology,
  Czech Republic
}

\begin{document}

\maketitle

\begin{abstract}
  We present the tool \ranker for complementing B\"{u}chi automata (BAs).
  \ranker builds on our previous optimizations of rank-based BA
  complementation and pushes them even further using numerous heuristics to
  produce even smaller automata.
  Moreover, it contains novel optimizations of specialized constructions for
  complementing
  (i)~inherently weak automata and
  (ii)~semi-deterministic automata, all delivered in a~robust tool.
  The optimizations significantly improve the usability of
  \ranker, as shown in an extensive experimental evaluation with
  real-world benchmarks, where \ranker produced in the majority of cases
  a~strictly smaller complement than other state-of-the-art tools.
\end{abstract}

\vspace{-2.0mm}
\section{Introduction}
\vspace{-0.0mm}

%
\noindent
B\"{u}chi automata (BA) complementation is an essential operation in the toolbox
of automata theory, logic, and formal methods.
It has many applications, e.g., implementing negation in decision procedures of
some logics (such as
the monadic second-order logic S1S~\cite{buchi1962decision,HLS-S1S},
the temporal logics EPTL and QPTL~\cite{sistla1987complementation}, or
the first-order logic over Sturmian words~\cite{pecan}),
proving termination of
programs~\cite{fogarty2009buchi,heizmann2014termination,ChenHLLTTZ18}, or model
checking of temporal properties~\cite{VardiW86}.
BA complementation also serves as the foundation stone of algorithms
for checking inclusion and equivalence of $\omega$-regular languages.
In all applications of BAs, the number of states of a~BA affects the overall
performance.
The many uses of BA complementation, as well as the challenging theoretical
nature of the problem, has incited researchers to develop a~number of different
approaches, e.g.,
\emph{determinization-based}~\cite{safra1988complexity,piterman2006nondeterministic,Redziejowski12},
\emph{rank-based}~\cite{KupfermanV01,FriedgutKV06,Schewe09}, or
\emph{Ramsey-based}~\cite{buchi1962decision,breuers-improved-ramsey},
some of them~\cite{fribourg,Schewe09} producing BAs with the number of states
asymptotically matching the lower bound $(0.76n)^n$ of Yan~\cite{yan}.
Despite their theoretical optimality, for many real-world
cases the constructions  create BAs with a~lot of unnecessary states, so
optimizations making the algorithms efficient in practice are
needed.

We present \ranker, a~robust tool for complementing (transition-based)~BAs.
\ranker uses several complementation approaches based on properties of the
input BA: it combines an optimization of the rank-based procedure developed
in~\cite{HavlenaL2021,HavlenaLS22,ChenHL19} with specialized (and further
optimized) procedures for complementing
semi-deterministic BAs~\cite{BlahoudekHSST16},
inherently weak BAs~\cite{Miyano84,BoigelotJW01}, and
elevator BAs~\cite{HavlenaLS22}.
An extensive experimental evaluation on a~wide range of automata
occurring in practice shows that \ranker can obtain a~smaller complement in
the majority of cases compared to the other state-of-the-art tools.

\renewcommand{\featcircle}{}
\renewcommand{\textfeatcircle}{}
\renewcommand{\postcircle}{}
\renewcommand{\textpostcircle}{}
\renewcommand{\precircle}{}
\renewcommand{\textprecircle}{}
\renewcommand{\iwacircle}{}
\renewcommand{\textiwacircle}{}
\renewcommand{\sdcircle}{}
\renewcommand{\textsdcircle}{}
\renewcommand{\gencircle}{}
\renewcommand{\textgencircle}{}

\paragraph{Contribution.}
We describe a~major improvement of
\ranker~\cite{HavlenaL2021,HavlenaLS22}, turning~it from a~prototype into a
robust tool.
We list
the particular optimizations below.
\begin{itemize}
	\item We extended the original BA complementation procedure
  with improved deelevation (cf.\ \cite{HavlenaLS22}) and advanced
  automata reductions.

  \item We also equipped \ranker with specialized constructions tailored for
    widely-used semi-deterministic and inherently weak automata.

  \item On top of that, we propose novel optimizations of the original NCSB
  construction for semi-deterministic BAs and a~simulation-based optimization of
    the Miyano-Hayashi algorithm for complementing inherently weak automata.
\end{itemize}
All of these improvements
are pushing the capabilities of \ranker, and also of practical BA complementation
itself, much further.


\vspace{-2.0mm}
\section{B\"{u}chi Automata}
\vspace{-0.0mm}

\paragraph{Words, functions.}
We fix a~finite nonempty alphabet~$\Sigma$ and the first infinite ordinal
$\omega = \{0, 1, \ldots\}$.
%
An (infinite) word~$\word$ is
a~function $\word\colon \omega \to \Sigma$ where the $i$-th
symbol is denoted as $\wordof i$. We~abuse notation and sometimes
represent~$\word$ as an~infinite sequence $\word = \wordof 0 \wordof 1 \dots$
$\Sigma^\omega$ denotes the set of all infinite words over~$\Sigma$.
%


\vspace{-1mm}
\paragraph{B\"{u}chi automata.}
A~(nondeterministic transition/state-based) \emph{B\"{u}chi
automaton} (BA) over~$\Sigma$ is a~quintuple $\aut = (Q, \trans, I,
\accstates, \acctrans)$ where $Q$ is a~finite set of \emph{states},
$\trans\colon Q \times \Sigma \to 2^Q$ is a~\emph{transition function}, $I
\subseteq Q$ is the sets of \emph{initial} states, and $\accstates \subseteq Q$
and $\acctrans \subseteq \trans$ are the sets of \emph{accepting states} and
\emph{accepting transitions} respectively. $\aut$ is called deterministic if
$|I|\leq 1$ and $|\delta(q,a)|\leq 1$ for each $q\in Q$ and $a \in \Sigma$.
We sometimes treat~$\trans$ as a~set of transitions $p \ltr a q$, for instance,
we use $p \ltr a q \in \trans$ to denote that $q \in \trans(p, a)$. Moreover, we
extend $\trans$ to sets of states $P \subseteq Q$ as $\trans(P, a) = \bigcup_{p
\in P} \trans(p,a)$.
The notation $\restrof \trans S$ for $S \subseteq Q$ is used to denote the
restriction of the transition function $\trans \cap (S \times \Sigma \times S)$.
Moreover, for $q \in Q$, we use $\aut[q]$ to denote the automaton $(Q, \trans, \{q\},
\accstates, \acctrans)$.

A~\emph{run}
of~$\aut$ from~$q \in Q$ on an input word $\word$ is an infinite sequence $\rho\colon
\omega \to Q$ that starts in~$q$ and respects~$\trans$, i.e., $\rho_0 = q$ and
$\forall i \geq 0\colon \rho_i \ltr{\wordof i}\rho_{i+1} \in \trans$.
Let $\infofst \rho \subseteq Q \cup \delta$ denote the set of states and transitions occurring in~$\rho$ infinitely often.
%
The run~$\rho$ is called \emph{accepting} iff $\infofst
\rho \cap (\accstates \cup \acctrans) \neq \emptyset$.
A~word~$\word$ is \emph{accepted by~$\aut$ from a~state~$q \in Q$} if $\aut$ has an
accepting run~$\rho$ on~$\word$ from~$q$, i.e., $\rho_0 = q$.
The set
$\langautof{\aut} q = \{\word \in \Sigma^\omega \mid \aut \text{ accepts } \word
\text{ from } q\}$ is called the \emph{language} of~$q$ (in~$\aut$). Given a~set
of states~$R \subseteq Q$, we define the language of~$R$ as $\langautof \aut R =
\bigcup_{q \in R} \langautof \aut q$ and the language of~$\aut$ as~$\langof \aut =
\langautof \aut I$.
If $\acctrans = \emptyset$, we call~$\aut$ \emph{state-based} and
if $\accstates = \emptyset$, we call~$\aut$ \emph{transition-based}.

A~\emph{co-B\"{u}chi automaton} (co-BA)~$\autc$ is the same as a~BA
except the definition of when a~run is accepting: a~run~$\rho$ of~$\autc$ is
\emph{accepting} iff $\infofst \rho \cap (\accstates \cup \acctrans) = \emptyset$.

\vspace{-1mm}
\paragraph{Automata types.}
Let $\aut = (Q, \trans, I, \accstates, \acctrans)$ be a BA. $C \subseteq Q$ is
a~\emph{strongly connected component} (SCC) of~$\aut$ if for any pair of states
$q, q' \in C$ it holds that~$q$ is reachable from~$q'$ and~$q'$ is reachable
from~$q$. $C$~is \emph{maximal} (MSCC) if it is not a~proper subset of another
SCC. An MSCC is \emph{non-accepting} if it contains no accepting state and no
accepting transition.
We say that an SCC~$C$ is \emph{inherently weak accepting} (IWA) iff \emph{every
cycle} in the transition diagram of~$\aut$ restricted to~$C$ contains an
accepting state or an accepting transition.
We say that an SCC~$C$ is \emph{deterministic} iff $(C, \restrof \trans C,
\emptyset, \emptyset, \emptyset)$ is deterministic.
$\aut$ is \emph{inherently weak} (IW) if all its MSCCs are inherently weak accepting
or non-accepting, and \emph{weak} if for states $q, q'$ that belong to the same
SCC, $q \in \accstates$ iff $q' \in \accstates$.
$\aut$ is \emph{semi-deterministic} (SDBA) if $\aut[q]$ is deterministic for every
$q \in \accstates \cup \{p \in Q \mid s \ltr a p \in \acctrans, s \in Q, a \in
\Sigma\}$.
Finally, $\aut$ is called \emph{elevator} if all its MSCCs are
inherently weak accepting, deterministic, or non-accepting.

\vspace{-2.0mm}
\section{Architecture}
\vspace{-0.0mm}

\ranker~\cite{ranker} is a~publicly available command line tool, written in C++,
implementing several approaches for complementation of (transition/state-based)
B\"{u}chi automata.
As an input,
\ranker accepts BAs in the HOA~\cite{BabiakBDKKM0S15} or
the simpler \texttt{ba}~\cite{abdulla2010simulation} format.
The architecture overview
is shown in Fig.~\ref{fig:ranker-arch}.
An input automaton is first adjusted by various
structural preprocessing steps to an intermediate equivalent automaton with a
form suitable for a complementation procedure. Based on the intermediate
automaton type, a concrete complementation procedure is used.
The result of the complementation is
subsequently polished by postprocessing steps, yielding an automaton on the
output. In the following text, we provide details about the internal blocks of
\ranker's architecture.

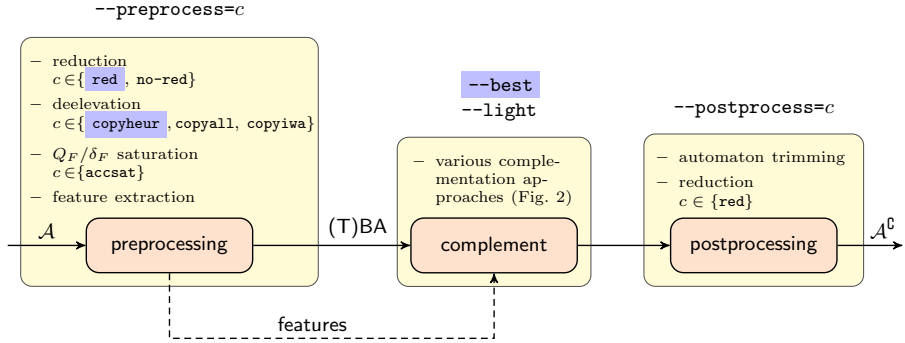
\begin{figure}[t]
	\centering
  \resizebox{\textwidth}{!}{
		\begin{tikzpicture}[->,>=stealth',scale=1,transform shape,semithick]
\tikzstyle{every state}=[inner sep=3pt,minimum size=5pt]

  \tikzstyle{decision}=[draw, diamond, aspect=2, inner sep=1pt,fill=Yellow!50]
  \tikzstyle{box}=[draw,rounded corners=2mm, inner sep=8pt,text width=2cm,align=center, node distance=4cm,fill=Orange!20]

  \node[box] (pre) {\textsf{preprocessing}};
  \node[box, right of=pre, xshift=10mm, yshift=0mm] (comp2) {\textsf{complement}};
  \node[box, right of=comp2] (post) {\textsf{postprocessing}};

  \path (pre.south east)[xshift=1.8mm,yshift=-0.5mm] \precircle;
  \path (post.south east)[xshift=1.8mm,yshift=-0.5mm] \postcircle;

  \node[above of=pre,text width=4.8cm, align=center, node distance=1.8cm, xshift=0mm] (pre-text) {
    \vbox{\scriptsize{\begin{itemize}
      \item reduction \\$c\,{\in}\{${\colorbox{blue!25!white}{\texttt{red}}}, \texttt{no-red}$\}$
      \smallskip
      \item deelevation\\$c\,{\in}\{${\colorbox{blue!25!white}{\texttt{copyheur}}},\,\texttt{copyall}, \texttt{copyiwa}$\}$
      \smallskip
      \item $\accstates$/$\acctrans$ saturation \\$c\,{\in}\{$\texttt{accsat}$\}$
      \smallskip
      \item feature extraction
    \end{itemize}}}
  };

  \node[above of=comp2,text width=3.0cm, align=center, node distance=1cm, xshift=0mm] (comp-text) {
    \vbox{\scriptsize{\begin{itemize}
      \item various complementation approaches (Fig.~\ref{fig:compl-scheme})
    \end{itemize}}}
  };

  \node[above of=post,text width=3.5cm, align=center, node distance=1.0cm, xshift=0mm] (post-text) {
    \vbox{\scriptsize{\begin{itemize}
      \item automaton trimming
      \smallskip
      \item reduction \\$c\in\{$\texttt{red}$\}$
    \end{itemize}}}
  };

  \node[above of=pre-text, node distance=1.8cm] (cmd-pre) {\texttt{--preprocess=}$c$};
  \node[above of=comp-text, node distance=1.3cm, text width=1cm] (cmd-comp) {{\colorbox{blue!25!white}{\texttt{--best}}} \texttt{--light}};
  \node[above of=post-text, node distance=1.1cm] (cmd-post) {\texttt{--postprocess=}$c$};

  \draw[->] (pre.east) -- node[above,xshift=4mm] {\textsf{(T)BA}} (comp2.west);
  \draw[->] (comp2.east) -- (post.west);
  \draw[->] ([xshift=-12mm]pre.west) -- node[above] {$\aut$} (pre.west);
  \draw[->] (post.east) -- node[above,xshift=2mm] {$\aut^\complement$} ([xshift=10mm]post.east);
  \draw[->, densely dashed] (pre.south) |- node[above, xshift=22mm,solid] (features) {\textsf{features}} ([yshift=-10mm]comp2.south) -- (comp2.south);

  \path (features)[xshift=10mm,yshift=-1mm] \featcircle;

  \begin{pgfonlayer}{background}
    \draw[-,very thin,rectangle,fill=Yellow!20,draw=black!70,rounded corners=5pt,inner sep=2pt]
      ([xshift=-10mm]pre.west) |- ([yshift=-2mm]pre-text.north) -| ([xshift=10mm,yshift=0mm]pre.east) |- ([yshift=-2mm]pre.south west) -| ([xshift=-10mm]pre.west);
    \draw[-,very thin,rectangle,fill=Yellow!20,draw=black!70,rounded corners=5pt,inner sep=2pt]
      ([xshift=-2mm]comp2.west) |- ([yshift=-2mm]comp-text.north) -| ([xshift=2mm,yshift=0mm]comp2.east) |- ([yshift=-2mm]comp2.south) -| ([xshift=-2mm]comp2.west);
    \draw[-,very thin,rectangle,fill=Yellow!20,draw=black!70,rounded corners=5pt,inner sep=2pt]
      ([xshift=-4mm]post.west) |- ([yshift=-2mm]post-text.north) -| ([xshift=4mm,yshift=0mm]post.east) |- ([yshift=-2mm]post.south west) -| ([xshift=-4mm]post.west);
  \end{pgfonlayer}

\end{tikzpicture}
  }
	\caption{
    Overview of the architecture of \ranker with the most important command-line
    options.
    Default settings are highlighted in blue.
	}
	\label{fig:ranker-arch}
\end{figure}

\vspace{-0.0mm}
\subsection{Preprocessing and Postprocessing}
\vspace{-0.0mm}

Before an input BA is sent to the complementation block itself, it is first
transformed into a~form most suitable for a concrete complementation technique.
On top of
that as a part of preprocessing, we identify structural features that are
further used to enabling/disabling certain optimizations during the
complementation.
After the complementation, the resulting automaton is optionally
reduced in a postprocessing step.
\ranker provides several options of preprocessing/postprocessing that are
discussed below.

\medskip
\noindent
\emph{Preprocessing}\textprecircle.
The following are the most important settings for preprocessing:

\begin{itemize}
  \item \emph{Reduction}: In order to obtain a~smaller automaton, reduction
    using \emph{direct simulation}~\cite{MayrC13} can be applied
    (\texttt{--preprocess=red}).
    Moreover, if the input automaton is IW or SDBA, we transform it into
    a~transition-based BA, which might be smaller (we only do local
    modifications and merge two states if they have the same successors while
    moving the acceptance condition from states to transitions entering
    accepting states).
    We, however, do not use this strategy for other BAs, because despite their
    possibly more compact representation, this reduction limits the effect of some
    optimizations used in the rank-based complementation procedure (the presence
    of accepting states allows to decrease the rank bound, cf.\ \cite{HavlenaLS22}).

  \item \emph{Deelevation}~\cite{HavlenaLS22}: For elevator automata,
    \ranker supports a~couple of deelevation strategies (extending a~basic
    version introduced in~\cite{HavlenaLS22}).
    Roughly speaking, deelevation makes a~copy of MSCCs such that each
    copied MSCC becomes a~terminal component (i.e., no run can leave it) and
    accepting states/transitions are removed from the original
    component (we call this the \emph{deelevation} of the
    component).
    Deelevation increases the number of states but decreases the rank bounds for
    rank-based complementation.
    \ranker offers several strategies that differ on which components are
    deelevated:
    \begin{itemize}
      \item  \texttt{--preprocess=copyall}: Every component is deelevated.
      \item  \texttt{--preprocess=copyiwa}: Only IWA components are deelevated.
      \item  \texttt{--preprocess=copyheur}: This option combines two
        modifications applied in sequence:
        \begin{inparaenum}[(i)]
          \item  If the input BA is not IW and the rank bound
            estimation~\cite{HavlenaLS22} of the BA is at least~5, then all
            MSCCs with an accepting state/transition are deelevated (the higher
            rank bound indicates a~longer sequence of components, for which
            deelevation is likely to be benefical).
          \item If on all paths from all initial states of the intermediate BA,
            the first non-trivial MSCC is non-accepting, then we partially
            determinize the initial part of the BA (up to the first non-trivial
            MSCCs); this reduces sizes of macrostates obtained in rank-based
            complementation.
        \end{inparaenum}
    \end{itemize}
  \item  \emph{Saturation of accepting states/transitions}:
    Since a~higher number of accepting states and transitions can help the
    rank-based complementation procedure, \ranker can (using
    \texttt{--preprocess=accsat}) saturate accepting states/transitions in the
    input BA (while preserving the language).
    This is, however, not always beneficial; for instance, saturation
    can break the structure for elevator rank estimation
    (cf.~\cite{HavlenaLS22}).
  \item  \emph{Feature extraction}\textfeatcircle:
    During preprocessing, we extract features of the BA that can help the complementation
    procedure in the second step.
    The features are, e.g., the type of the BA, rank bounds for individual
    states~\cite{HavlenaLS22}, or settings of particular optimizations
    from~\cite{HavlenaL2021} (e.g., for deterministic automata with
    a~smaller rank bound, it is counter-productive to use techniques reducing
    the rank bound based on reasoning about the waiting part).
\end{itemize}

\medskip
\noindent
\emph{Postprocessing}\textpostcircle.
After the complementation procedure finishes, \ranker removes useless states and
optionally applies simulation reduction (\texttt{--postprocess=red}).

\vspace{-2.0mm}
\subsection{Complementation Approaches}
\vspace{-0.0mm}

Based on the automaton type,
\ranker uses several approaches for complementation (cf.\
\cref{fig:compl-scheme}).
These are, ordered by decreasing priority, the following:
%
\begin{figure}[t]
	\centering
	\scalebox{0.8}{
		\begin{tikzpicture}[->,>=stealth',scale=1,transform shape,semithick]
\tikzstyle{every state}=[inner sep=3pt,minimum size=5pt]

  \tikzstyle{decision}=[draw, diamond, aspect=2, inner sep=1pt,fill=Orange!20]
  \tikzstyle{box}=[draw,rounded corners=2mm, inner sep=0pt,text width=6.2cm,align=center, node distance=1.75cm,fill=Yellow!20]

  \node[box] (iw) {
    \vbox{\scriptsize{\begin{itemize}
      \item Miyano-Hayashi construction~\cite{Miyano84}
      \item Macrostates simulation-based pruning/saturation optimization (Sec.~\ref{sec:iw-sim-opt})
    \end{itemize}}}
  };
  \node[box,below of=iw] (sd) {
    \vbox{\scriptsize{\begin{itemize}
      \item \ncsblazy construction~\cite{ChenHLLTTZ18}
      \item \ncsbmaxrank construction (Sec.~\ref{sec:ncsb-maxrank})
      \item Optimized Rank-based construction~\cite{HavlenaL2021,HavlenaLS22}
    \end{itemize}}}
  };
  \node[box,below of=sd, node distance=1.6cm] (gen) {
    \vbox{\scriptsize{\begin{itemize}
      \item Optimized Rank-based construction~\cite{HavlenaL2021,HavlenaLS22}
      \item Backoff: \spot~\cite{HavlenaL2021}
    \end{itemize}}}
  };

  \node[draw=black,circle,inner sep=2pt,right of=sd, node distance=4cm] (join) {};
  \node[inner sep=1pt,right of=join, node distance=10mm] (join2) {};
  \node[decision,left of=sd, node distance=5.5cm] (type) {\textsf{type}};

  \path (iw.south east) [xshift=2mm,yshift=-0.5mm]\iwacircle;
  \path (sd.south east) [xshift=2mm,yshift=-0.5mm] \sdcircle;
  \path (gen.south east) [xshift=2mm,yshift=-0.5mm] \gencircle;

  %
  %

  \draw[->] (type.north) |- node[pos=0.74,above] {\textsf{inherently weak}} (iw.west);
  \draw[->] (type.east) -- node[pos=0.4,above] {\textsf{SDBA}} (sd.west);
  \draw[->] (type.south) |- node[pos=0.73,below] {\textsf{otherwise}} (gen.west);
  \draw[->] ([xshift=-4mm]type.west) -- (type.west);

  \draw[->] (sd.east) -- (join);
  \draw[->] (iw.east) -| (join);
  \draw[->] (gen.east) -| (join);
  \draw[->] (join) -- (join2);


\end{tikzpicture}
	}
  \vspace{-0mm}
	\caption{
	Overview of complementation approaches used in \ranker.
	}
	\label{fig:compl-scheme}
  \vspace*{-2mm}
\end{figure}
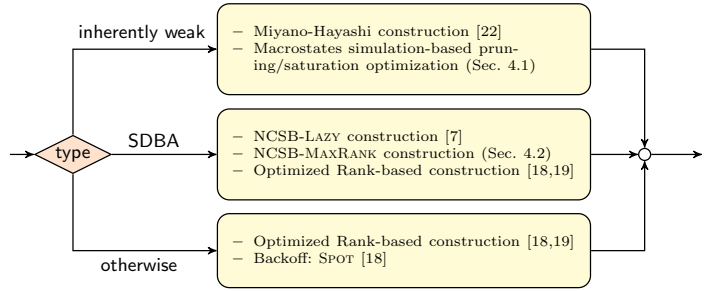

\begin{itemize}
  \item  \emph{Inherently weak BAs}\textiwacircle:
    For the complementation of inherently weak automata, both the Miyano-Hayashi
    construction~\cite{Miyano84} and its optimization of adjusting macro\-states
    (described in \cref{sec:iw-sim-opt}), are implemented.
    The construction converts an input automaton into an intermediate equivalent
    co-B\"{u}chi automaton, which is then complemented.
    The implemented optimizations adjust macrostates of the Miyano-Hayashi
    construction according to a~direct simulation relation.
    By default (\texttt{--best}), the Miyano-Hayashi construction and the
    optimization of pruning simulation-smaller states from macrostates are used
    and the smaller result is output.
    For the option \texttt{--light}, only the optimized construction is used.

  \item  \emph{Semi-deterministic BA}\textsdcircle:
    For SDBAs, \ranker by default (\texttt{--best}) uses both an
    NCSB-based~\cite{BlahoudekHSST16} procedure and an optimized rank-based
    construction with advanced rank estimation~\cite{HavlenaL2021,HavlenaLS22};
    the smaller result is picked.
    The particular NCSB-based procedure used is \ncsbmaxrank from
    \cref{sec:ncsb-maxrank} (\ranker also contains an implementation of
    \ncsblazy from~\cite{ChenHLLTTZ18}, which can be turned on using
    \texttt{--ncsb-lazy}, but usually gives worse results).
    For the option \texttt{--light}, only \ncsbmaxrank is used.

  \item  \emph{Otherwise}\textgencircle:
    For BAs with no special structure, \ranker uses the optimized
    rank-based complementation algorithm from~\cite{HavlenaL2021,HavlenaLS22}
    with \spot as the backoff~\cite{HavlenaL2021} (i.e., \ranker can determine
    when the input has a~structure that is bad for the rank-based procedure and
    use another approach).
    Particular optimizations are selected according to the features of the input
    BA (e.g., the number of states or the structure of the automaton).
\end{itemize}

\vspace{-2.0mm}
\section{Optimizations of the Constructions}
\label{sec:optimizations}
\vspace{-0.0mm}

In this section, we provide details about new optimizations of complementation
of inherently weak and semi-deterministic automata implemented in \ranker.
Proofs of their correctness can be found in the Appendix.

%

\vspace{-2.0mm}
\subsection{Macrostates Adjustment for Inherently Weak Automata}
\label{sec:iw-sim-opt}
\vspace{-0.0mm}

For complementing IW automata, \ranker uses a~method based on the Miyano-Hayashi
construction (denoted as \mihay)~\cite{Miyano84}:
In the first step, accepting states of an input IW BA~$\aut$ are saturated to
obtain a~language-equivalent weak automaton $\autw = (Q, \delta, I, \accstates,
\emptyset)$ (we remove accepting transitions because they do not provide any
advantage for IW automata).
In the second step, $\autw$ is converted to the
equivalent co-B\"{u}chi automaton $\autc = (Q, \delta, I, \accstates' =
Q\setminus \accstates, \emptyset)$ by swapping accepting and non-accepting
states.
Finally, the Miyano-Hayashi construction is used to obtain the complement
(state-based) BA.



Our optimizations of the \mihay procedure are inspired by optimizations of the
determinization algorithm for automata over finite words~\cite{GlabbeekP08} and
by saturation of macrostates in rank-based BA complementation
procedure~\cite{ChenHL19}, where simulation relations are used to adjust
macrostates in order to obtain a~smaller automaton.
We modify the original construction by introducing an \emph{adjustment function}
that modifies obtained macrostates, either to obtain \emph{smaller} macrostates
(for \emph{pruning} strategy) or \emph{larger} macrostates (for
\emph{saturating} strategy; the hope is that \emph{more} original macrostates
map to \emph{the same} saturated macrostate).
Formally, given a~co-BA~$\autc$ and an \emph{adjustment
function} $\theta\colon 2^Q \to 2^Q$, the construction $\cobapar$ gives the
(deterministic, state-based) BA $\cobapar(\autc) = (Q', \delta', I',
\accstates', \emptyset)$, whose components are defined as follows:
\begin{itemize}
	\item $Q' = 2^Q \times 2^Q$,
	\item $I' = \{(\theta(I), \theta(I) \setminus \accstates')\}$,
	\item $\delta'((S, B), a) = (S', B')$ where
	\begin{itemize}
		\item $S' = \theta(\delta(S, a))$,
		\item and
		\begin{itemize}
			\item $B' = S' \setminus \accstates'$ if $B = \emptyset$ or
			\item $B' = (\delta(B, a) \cap S') \setminus \accstates'$ if $B \neq \emptyset$, and
		\end{itemize}
	\end{itemize}
	\item $F' = 2^Q \times \{\emptyset\}$.
\end{itemize}
Intuitively, the construction tracks in the~$S$-component all runs over a~word
and uses the~$B$-component to check that each of the runs sees infinitely many
accepting states from~$\accstates'$ (by a~cut-point construction).
The original \mihay procedure can be obtained by using identity for the
adjustment function, $\theta = \mathrm{id}$.



In the following, we use~$\dirsimbyw$
and~$\fairsimbyc$ to denote a~\emph{direct simulation} on~$\autw$ and
a~\emph{fair simulation} on~$\autc$ respectively (see,
e.g.,~\cite{etessami2002hierarchy} for more details; in particular, $p
\fairsimbyc q$ iff for every trace of~$\autc$ from state~$p$ over~$\word$
with finitely many accepting states, there exists a~trace from~$q$ with
finitely many accepting states over~$\word$).

Let ${\sqsubseteq}\subseteq Q\times Q$ be a~relation on the states of~$\autc$
defined as follows:
$p \sqsubseteq q$ iff
\begin{inparaenum}[(i)]
	\item $p \fairsimbyc q$,
	\item $q$ is reachable from~$p$ in~$\autc$, and
	\item either $p$ is not reachable from~$q$ in~$\autc$ or $p=q$.
\end{inparaenum}
The two adjustment functions $\pr,\sat\colon 2^Q \to 2^Q$ are then defined for each
$S \subseteq Q$ as follows:
\begin{itemize}
  \item  \emph{pruning}:
      $\pr(S) =S'$ where $S'\subseteq S$ is the lexicographically smallest set
      (given a~fixed ordering on~$Q$)
      such that $\forall q \in S \exists q'\in S' \colon q\sqsubseteq q'$ and
  \item  \emph{saturating}:
    $\sat(S) = \{ p\in Q \mid \exists q \in Q\colon p \fairsimbyc q \}$.
\end{itemize}
%
Informally, $\pr$ removes simulation-smaller states and $\sat$ saturates
a~macrostate with all simulation-smaller states.\footnote{%
It has been brought to our attention by Alexandre Duret-Lutz that a~strategy
similar to \emph{pruning} with direct simulation has been implemented in
\spot's~\cite{spot} determinization and, moreover, generalized
in~\cite{LodingP19} to also work in some cases \emph{within} SCCs.
}
%
%
The correctness of the constructions is summarized by the following theorem:
%
\begin{restatable}{theorem}{theCoBaCorr}\label{the:coba-corr}
	For a co-BA $\autc$, $\langof{\cobasat(\autc)} =
	\langof{\cobapr(\autc)} = \Sigma^\omega \setminus \langof{\autc}$.
\end{restatable}

\noindent
In \ranker, we approximate a~fair simulation~$\fairsimbyc$ by a~direct
simulation~$\dirsimbyw$ (which is easier to compute); the correctness holds
due to the following lemma:

\begin{restatable}{lemma}{theFairSim}
  Let $\autw = (Q, \delta, I, \accstates, \emptyset)$ be a~weak BA and
  $\autc = (Q, \delta, I, \accstates' = Q\setminus \accstates, \emptyset)$ be
  a~co-BA.
  Then $\dirsimbyw {\subseteq} \fairsimbyc$.
\end{restatable}


\vspace{-0.0mm}
\subsection{\ncsbmaxrank Construction}
\label{sec:ncsb-maxrank}
\vspace{-0.0mm}


The structure of semi-deterministic BAs allows to use more efficient
complementation techniques.
From the point of view of rank-based complementation,
the maximum rank of semi-deterministic automata can be bounded by~3.
If a~rank-based complementation procedure based on \emph{tight rankings} (such
as~\cite{HavlenaL2021,HavlenaLS22}) is used to complement an SDBA, it can suffer
from having too many states due to the presence of the \emph{waiting} part
(intuitively, runs wait in the waiting part of the complement until they can see
only tight rankings, then they jump to the \emph{tight} part where they can
accept, cf.~\cite{FriedgutKV06,Schewe09,HavlenaL2021} for more details).
Furthermore, the information about ranks of individual runs may sometimes be
more precise than necessary, which disables merging some runs.
The NCSB construction~\cite{BlahoudekHSST16} overcomes these issues by not
considering the waiting part and keeping only rough information about the ranks.
As a~matter of fact, NCSB and the rank-based approach are not
comparable due to tight-rankings and additional techniques restricting the
ranking functions~\cite{HavlenaL2021,HavlenaLS22}, taking into account
structural properties of the automaton, which is why \ranker in the default
setting tries both rank-based and NCSB-based procedures for complementing SDBAs.


An issue of the NCSB algorithm is a~high degree of nondeterminism of the
constructed BA (and therefore also a~higher number of states).
The \ncsblazy construction~\cite{ChenHLLTTZ18} improves the original
algorithm with postponing the nondeterministic choices, which usually produces
smaller results.
Even the \ncsblazy construction may, however, suffer in some cases from
generating too many successors.
We propose an
improvement of the original NCSB algorithm, inspired by the \algmaxrank
construction in rank-based complementation from~\cite{HavlenaL2021} (which is
inspired by~\cite[Section~4]{Schewe09}), hence called the \ncsbmaxrank
construction, reducing the number of successors of any macrostate and symbol to
at most two.


Formally, for a given SDBA $\aut = (\stnondet \uplus \stdet, \delta = \delnondet
\uplus \deldet \uplus \deltrans, I, \accstates, \acctrans)$ where~$\stdet$ are
the states reachable from an accepting state or transition and
$\stnondet$ is the rest, $\delnondet =\delta_{|\stnondet} $, $\deldet = \delta_{|\stdet}$,
and $\deltrans$ is the transition function between $\stnondet$ and~$\stdet$, we
define
$\ncsbmaxrank(\aut) = (Q',I',\delta',\accstates', \emptyset)$ to be the
(state-based) BA whose components are the following:
\begin{itemize}
	\item $Q' = \{ (N,C,S,B) \in 2^{\stnondet}\times 2^{\stdet}\times 2^{\stdet\setminus \accstates}\times 2^{\stdet} \mid B \subseteq C \}$,
	\item $I' = \{ (\stnondet\cap I, \stdet\cap I, \emptyset, \stdet \cap I) \}$,
  \item $\delta' = \gamma_1 \cup \gamma_2$ where the successors of a~macrostate
    $(N,C,S,B)$ over $a \in \Sigma$ are defined such that if $\acctrans(S,a)
    \neq \emptyset$ then $\delta'((N, C, S, B), a) = \emptyset$, else
    \begin{itemize}
      \item $\gamma_1((N, C, S, B), a) = \{(N', C', S', B')\}$ where
        \begin{itemize}
          \item $N' = \delnondet(N,a)$,
          \item $S' = \deldet(S, a)$,
          \item $C' = (\deltrans(N,a) \cup \deldet(C, a))\setminus S'$, and
          \item $B' = C'$ if $B = \emptyset$, otherwise $B' = \deldet(B, a) \cap C'$.
        \end{itemize}
      \item If $B' \cap \accstates = \emptyset$, we also set $\gamma_2((N, C, S, B), a) =
        \{(N', C^\bullet, S^\bullet, B^\bullet)\}$ with
        \begin{itemize}
          \item $B^\bullet = \emptyset$,
          \item $S^\bullet = S' \cup B'$, and
          \item $C^\bullet = C' \setminus S^\bullet$,
        \end{itemize}
        else $\gamma_2((N, C, S, B), a) = \emptyset$.
    \end{itemize}
	\item $\accstates' = \{(N,C,S,B) \in Q' \mid B = \emptyset\}$.
\end{itemize}
Intuitively,
\ncsbmaxrank provides at most two choices for each macrostate:
either keep all states in $B$ or move all states from $B$ to $S$ (if $B$
contains no accepting state). If a word is not accepted by $\aut$, it
will be safe to put all states from $B$ to $S$ at some point.
The construction is in fact incomparable to the original NCSB
algorithm~\cite{BlahoudekHSST16} (in particular due to the condition $C'
\subseteq \deldet(C\setminus \accstates, a)$, which need not hold in
\ncsbmaxrank).
Correctness of the construction is given by the following theorem.

\begin{restatable}{theorem}{theMaxrankCorr}
	Let $\aut$ be an SDBA.
  Then $\langof{\ncsbmaxrank(\aut)} = \Sigma^\omega \setminus \langof{\aut}$.
\end{restatable}

\vspace{-2mm}
\section{Experimental Evaluation}
\vspace{-0mm}
%
%
%
%

\newcommand{\figoldspot}[0]{
\begin{figure}[t]
  \vspace*{-2mm}
  \begin{subfigure}[b]{0.49\linewidth}
  \begin{center}
  \includegraphics[width=\linewidth,keepaspectratio]{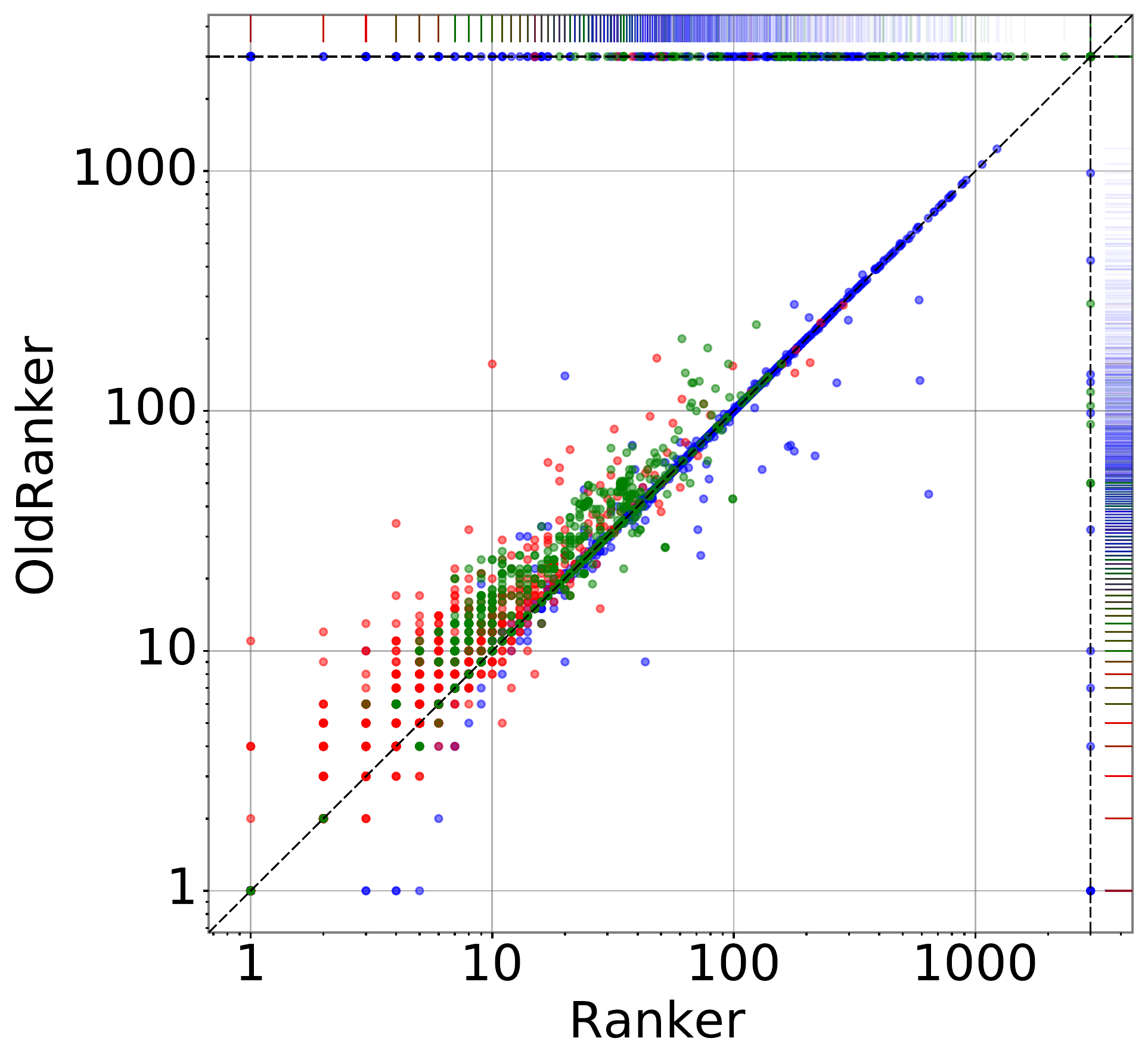}
  \end{center}
  \vspace{-3mm}
  \caption{$\ranker$ vs $\rankerold$}
  \label{fig:rankerspot}
  \end{subfigure}
  \begin{subfigure}[b]{0.49\linewidth}
  \begin{center}
  \includegraphics[width=\linewidth,keepaspectratio]{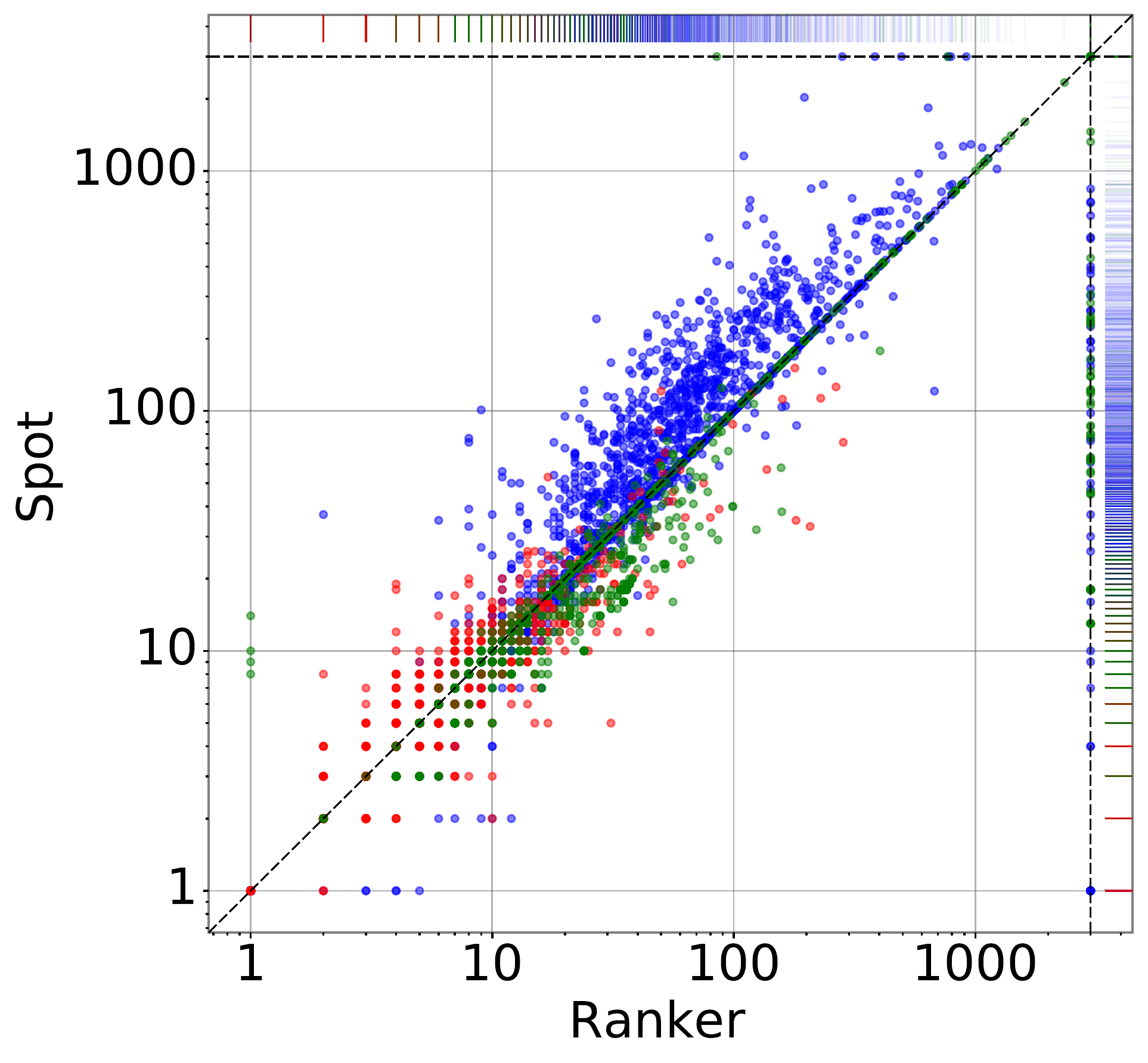}
  \end{center}
  \vspace{-3mm}
  \caption{$\ranker$ vs $\spot$}
  \label{fig:rankerroll}
  \end{subfigure}
  \vspace{0mm}
\caption{Comparison of the complement size obtained by \ranker,
  $\rankerold$, and \spot (horizontal and vertical dashed lines represent
  timeouts).}
\label{fig:comp-others}
\vspace*{-3mm}
\end{figure}
}

\newcommand{\figiwsdba}[0]{
\begin{figure}[t]
  \vspace*{-2mm}
  \begin{subfigure}[b]{0.49\linewidth}
  \begin{center}
  \includegraphics[width=\linewidth,keepaspectratio]{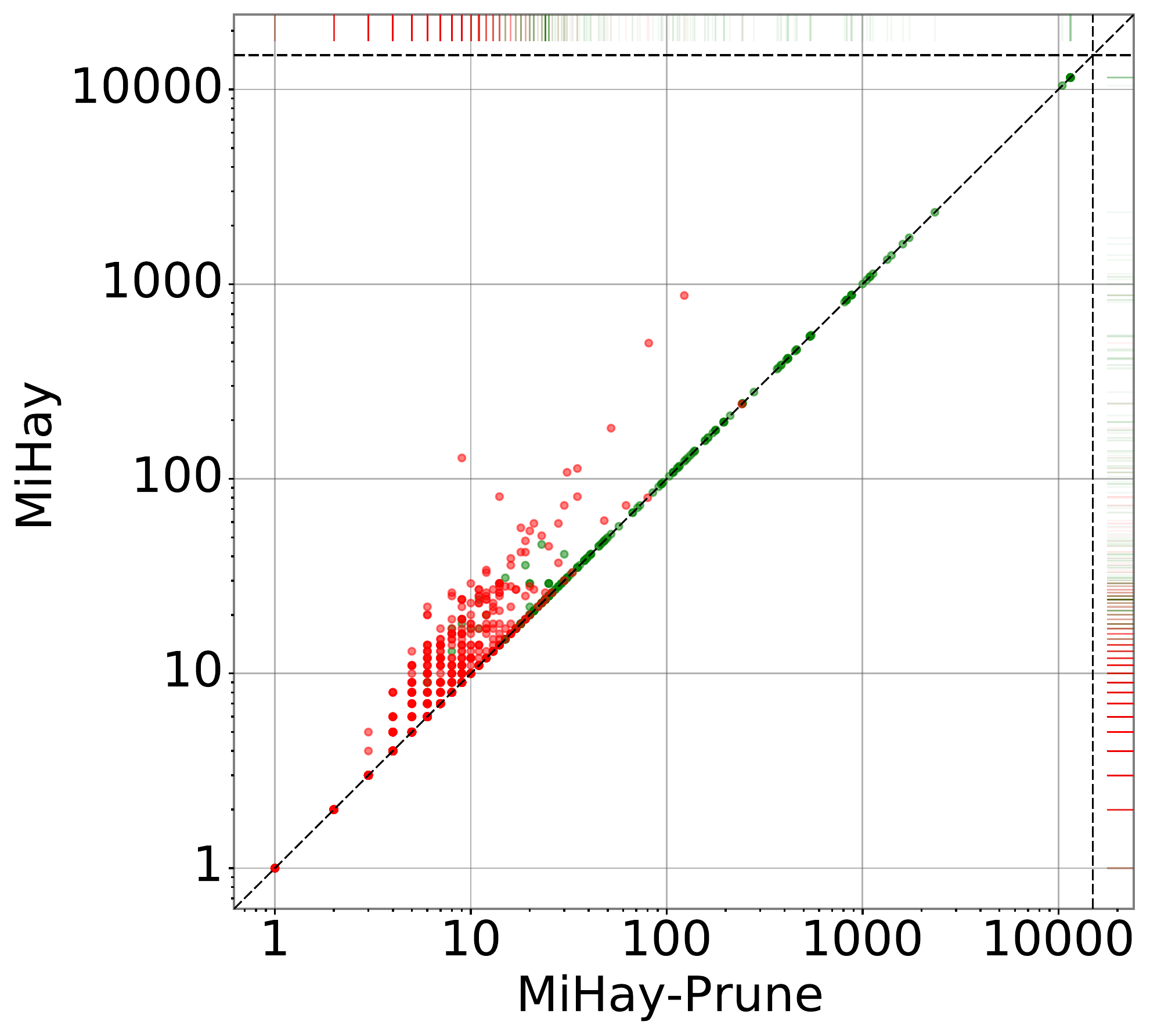}
  \end{center}
  \vspace{-3mm}
  \caption{$\cobapr$ vs $\mihay$}
  \label{fig:res-mihay}
  \end{subfigure}
  \begin{subfigure}[b]{0.49\linewidth}
  \begin{center}
  \includegraphics[width=\linewidth,keepaspectratio]{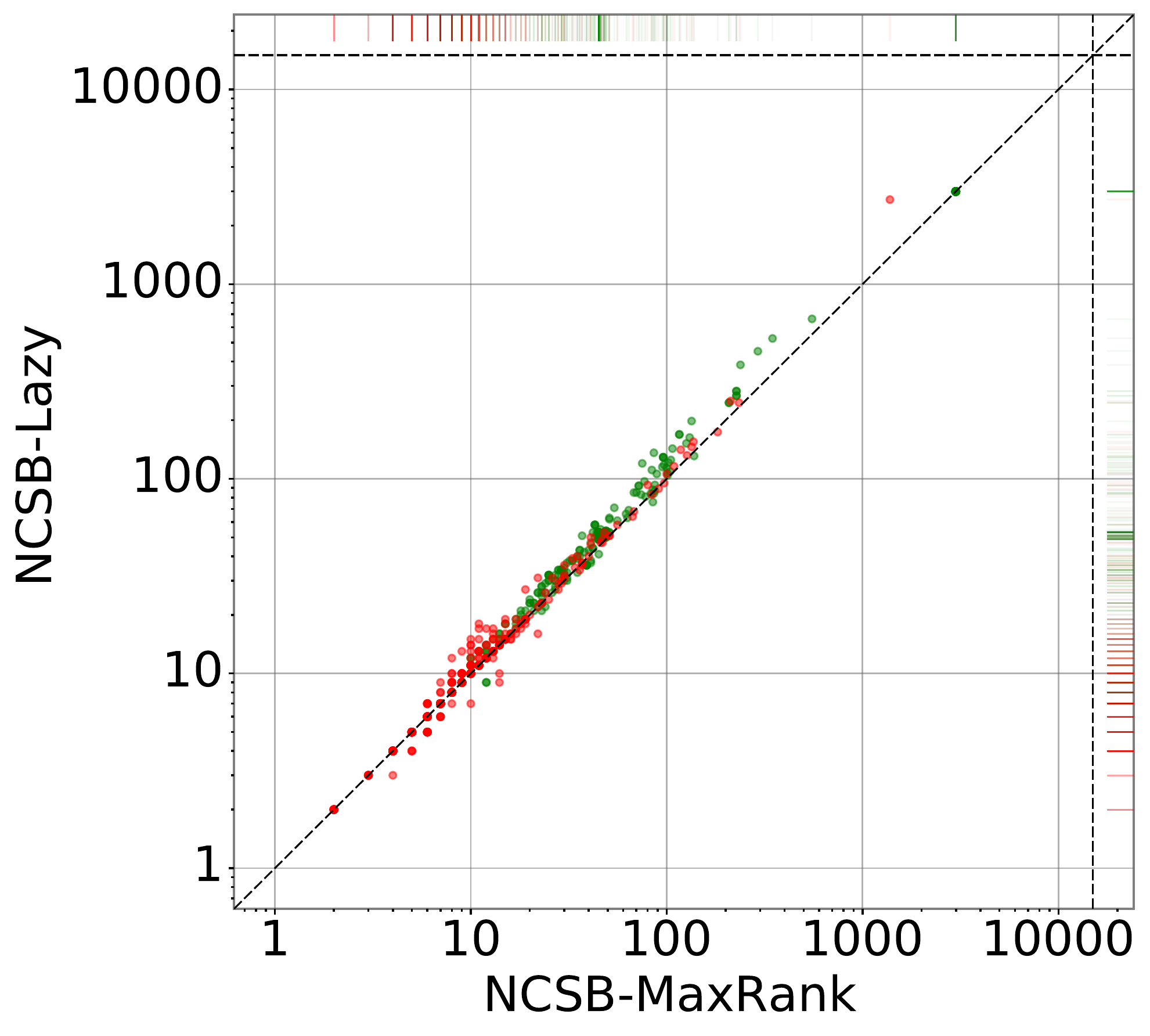}
  \end{center}
  \vspace{-3mm}
  \caption{$\ncsbmaxrank$ vs $\ncsblazy$}
  \label{fig:res-ncsb}
  \end{subfigure}
  \vspace{0mm}
\caption{Evaluation of the effect of our optimizations for IW and SDBA automata.
  }
\label{fig:comp-others}
\end{figure}
}

\newcommand{\tabresults}[0]{
  \begin{table}[t]
  \caption{Statistics for our experiments.
    The table compares the sizes of complement BAs obtained by \ranker and other
    approaches (after postprocessing).
    The \textbf{wins} and \textbf{losses} columns give the number of times when \ranker
    was strictly better and worse.
    The values are given for the three datasets as ``\dsall (\dsrandom{} :
    \dsltl{} : \dsautomizer)''.
    Approaches in \goal are labelled with~\goalmark.
    }
    \vspace{3mm}
  \label{tab:results}
  \hspace{-3mm}
  \resizebox{1.035\linewidth}{!}{%
    \newcolumntype{g}{>{\columncolor{Gray!30}}r}
\newcolumntype{f}{>{\columncolor{Gray!30}}l}
\newcolumntype{h}{>{\columncolor{Gray!30}}c}
\begin{tabular}{lgggggfrrrrrlgggggfrrrrrlgggggf}
\toprule
\multicolumn{1}{c}{\textbf{method}} & \multicolumn{6}{h}{\textbf{mean}}   & \multicolumn{6}{c}{\textbf{median}}   &   \multicolumn{6}{h}{\textbf{wins}} & \multicolumn{6}{c}{\textbf{losses}}   & \multicolumn{6}{h}{\textbf{timeouts}}   \\
\midrule
\rowcolor{GreenYellow}\ranker & 38 & (44 & \!:\! &  9 & \!:\! & 67) & 11 & (18 & \!:\! &  5 & \!:\! & 22) &      &       &       &      &       &      &      &      &       &     &       &      &                                   158 & (53                                 & \!:\!                                 &                                   0 & \!:\!                                 & 105)                                    \\
 $\rankerold$                  & 30 & (38 & \!:\! & 10 & \!:\! & 32) & 12 & (18 & \!:\! &  6 & \!:\! & 22) & 1554 & (356  & \!:\! & 650  & \!:\! & 548) & 264  & (142 & \!:\! & 69  & \!:\! & 53)  &                                   458 & (259                                & \!:\!                                 &                                   7 & \!:\!                                 & 192)                                    \\
 \piterman~\goalmark           & 43 & (56 & \!:\! & 12 & \!:\! & 38) & 14 & (19 & \!:\! &  8 & \!:\! & 24) & 2881 & (1279 & \!:\! & 966  & \!:\! & 636) & 392  & (263 & \!:\! & 68  & \!:\! & 61)  &                                   309 & (12                                 & \!:\!                                 &                                   4 & \!:\!                                 & 293)                                    \\
 \safra~\goalmark              & 49 & (60 & \!:\! & 17 & \!:\! & 56) & 15 & (18 & \!:\! & 10 & \!:\! & 24) & 3109 & (1348 & \!:\! & 1117 & \!:\! & 644) & 274  & (229 & \!:\! & 31  & \!:\! & 14)  &                                   599 & (160                                & \!:\!                                 &                                  30 & \!:\!                                 & 409)                                    \\
 \spot                         & 46 & (57 & \!:\! &  8 & \!:\! & 66) & 11 & (18 & \!:\! &  5 & \!:\! & 18) & 1347 & (935  & \!:\! & 339  & \!:\! & 73)  & 1057 & (327 & \!:\! & 343 & \!:\! & 387) &                                    73 & (13                                 & \!:\!                                 &                                   0 & \!:\!                                 & 60)                                     \\
 \fribourg~\goalmark           & 49 & (68 & \!:\! &  8 & \!:\! & 27) & 11 & (18 & \!:\! &  6 & \!:\! & 19) & 2223 & (1177 & \!:\! & 503  & \!:\! & 543) & 586  & (245 & \!:\! & 207 & \!:\! & 134) &                                   399 & (93                                 & \!:\!                                 &                                   2 & \!:\!                                 & 304)                                    \\
 \ltldstar                     & 44 & (56 & \!:\! & 12 & \!:\! & 47) & 14 & (19 & \!:\! &  7 & \!:\! & 24) & 2794 & (1297 & \!:\! & 924  & \!:\! & 573) & 448  & (283 & \!:\! & 88  & \!:\! & 77)  &                                   288 & (130                                & \!:\!                                 &                                  13 & \!:\!                                 & 145)                                    \\
 \seminator                    & 46 & (58 & \!:\! &  8 & \!:\! & 64) & 11 & (17 & \!:\! &  5 & \!:\! & 21) & 1626 & (1297 & \!:\! & 291  & \!:\! & 38)  & 1113 & (286 & \!:\! & 398 & \!:\! & 429) &                                   419 & (368                                & \!:\!                                 &                                   1 & \!:\!                                 & 50)                                     \\
 \roll                         & 18 & (15 & \!:\! & 11 & \!:\! & 54) &  9 & (8  & \!:\! &  8 & \!:\! & 28) & 6050 & (3824 & \!:\! & 1551 & \!:\! & 675) & 620  & (369 & \!:\! & 125 & \!:\! & 126) &                                  1893 & (1595                               & \!:\!                                 &                                   8 & \!:\!                                 & 290)                                    \\
\bottomrule
\end{tabular}

  }
  \vspace*{-3mm}
  \end{table}
}

\newcommand{\tabletimes}[0]{
\begin{wraptable}[11]{r}{6.3cm}
\vspace{-6mm}
\caption{Run times of the tools [s]
given as ``\dsall (\dsrandom{} : \dsltl{} : \dsautomizer)''}
\label{tab:times}
\vspace{-1mm}
\hspace*{-3mm}
\resizebox{6.6cm}{!}{
\newcolumntype{g}{>{\columncolor{Gray!30}}r}
\newcolumntype{f}{>{\columncolor{Gray!30}}l}
\newcolumntype{h}{>{\columncolor{Gray!30}}c}
\begin{tabular}{lgggggfrrrrrl}
\toprule
\multicolumn{1}{c}{\textbf{method}} & \multicolumn{6}{h}{\textbf{mean}}   & \multicolumn{6}{c}{\textbf{median}}   \\
\midrule
\rowcolor{GreenYellow}\ranker &  3.72 & (4.34  & \!:\! & 0.45 & \!:\! & 7.30)  & 0.05 & (0.10 & \!:\! &                                  0.04 & \!:\!                               & 0.08)                                 \\
 $\rankerold$                  &  4.62 & (5.33  & \!:\! & 0.72 & \!:\! & 9.69)  & 0.07 & (0.19 & \!:\! &                                  0.03 & \!:\!                               & 0.15)                                 \\
 \piterman~\goalmark           &  8.06 & (6.07  & \!:\! & 5.95 & \!:\! & 28.38) & 5.12 & (4.96 & \!:\! &                                  5.08 & \!:\!                               & 8.68)                                 \\
 \safra~\goalmark              & 11.58 & (10.41 & \!:\! & 6.51 & \!:\! & 38.65) & 5.41 & (5.32 & \!:\! &                                  5.26 & \!:\!                               & 9.02)                                 \\
 \spot                         &  0.64 & (0.57  & \!:\! & 0.02 & \!:\! & 2.28)  & 0.02 & (0.02 & \!:\! &                                  0.01 & \!:\!                               & 0.02)                                 \\
 \fribourg~\goalmark           & 13.13 & (14.14 & \!:\! & 6.06 & \!:\! & 23.88) & 5.69 & (6.82 & \!:\! &                                  4.92 & \!:\!                               & 6.57)                                 \\
 \ltldstar                     &  2.1  & (2.25  & \!:\! & 0.34 & \!:\! & 5.15)  & 0.02 & (0.02 & \!:\! &                                  0.01 & \!:\!                               & 0.05)                                 \\
 \seminator                    &  4.16 & (6.33  & \!:\! & 0.03 & \!:\! & 1.88)  & 0.03 & (0.08 & \!:\! &                                  0.01 & \!:\!                               & 0.03)                                 \\
 \roll                         & 23.65 & (29.82 & \!:\! & 3.88 & \!:\! & 49.02) & 3.34 & (6.19 & \!:\! &                                  1.71 & \!:\!                               & 17.14)                                \\
\bottomrule
\end{tabular}

}
\end{wraptable}
}

\figiwsdba

We compared the improved version of \ranker presented in this paper with other state-of-the-art tools, namely,
\goal~\cite{goal} (implementing \piterman~\cite{piterman2006nondeterministic}, \safra~\cite{safra1988complexity},
and \fribourg~\cite{fribourg}),
\spot~2.9.3~\cite{spot} (implementing Redziejowski's
algorithm~\cite{Redziejowski12}), \seminator~\cite{seminator},
\ltldstar~0.5.4~\cite{KleinB07}, \roll~\cite{roll}, and the previous version of
\ranker from~\cite{HavlenaLS22}, denoted as $\rankerold$.
All tools were set to the mode where they output a~state-based BA.
The correctness of our implementation was tested using \spot's
\texttt{autcross} on all of BAs from our benchmarks.
The experimental evaluation was performed on a~64-bit \textsc{GNU/Linux Debian}
workstation with an Intel(R) Xeon(R) CPU E5-2620 running at 2.40\,GHz with
32\,GiB of RAM, using a~5-minute timeout.
Axes in plots are logarithmic.
An~artifact that allows reproduction of the results is available as~\cite{Ranker22artifact}.

\paragraph{Datasets.}
We use automata from the following three datasets:
\begin{inparaenum}[(i)]
  \item  \dsrandom containing 11,000 BAs over a~two letter alphabet used in~\cite{tsai-compl}, which
    were randomly generated via the Tabakov-Vardi approach~\cite{TabakovV05},
    starting from 15 states and with various parameter settings;
  \item  \dsltl with 1,721 BAs over larger alphabets (up to 128 symbols) used in~\cite{seminator},
    obtained from LTL formulae from literature (221) or randomly generated (1,500),
  \item \dsautomizer containing 906 BAs over larger alphabets (up to $2^{35}$ symbols) used in~\cite{ChenHLLTTZ18}, which were obtained from the \automizer tool
(all benchmarks are available at~\cite{Lengal22AB}).
\end{inparaenum}
Note that we included \dsrandom in order to simulate applications that cannot
easily generate BAs of one of the easier fragments (unlike, e.g., \automizer,
which generates in most cases SDBAs) and have thus, so
far, not been seriously considered by the community due to the lack of
practically efficient BA complementation approaches (e.g., the automata-based
S1S decision procedure~\cite{buchi1962decision}).
All automata were preprocessed using \spot's \autfilt (using the \verb=--high=
simplification level),
and converted to the HOA format~\cite{BabiakBDKKM0S15}.
%
We~also removed trivial one-state BAs.
	In~the end, we were left with 4,533 (\dsrandom, \textcolor{blue}{\bf blue} data
  points), 1,716
	(\dsltl, \textcolor{red}{\bf red} data points), and 906 (\dsautomizer,
  \textcolor{green!70!black}{\bf green} data points) automata.
We use \dsall to denote their union (7,155 BAs).

\figoldspot

\vspace{-2mm}
\subsection{Effect of the Proposed Optimizations}
\vspace{-0mm}

\newcommand{\taboptim}[0]{
\begin{wraptable}[11]{r}{5.4cm}
\vspace*{3mm}
\caption{Effects of our optimizations for IW and SDBA automata.
  Sizes of output BAs are given as ``\dsboth (\dsltl{} : \dsautomizer)''.
  }
\label{tab:optim}
\vspace{-1mm}
\hspace*{-4mm}
\resizebox{5.9cm}{!}{
\newcolumntype{g}{>{\columncolor{Gray!30}}r}
\newcolumntype{f}{>{\columncolor{Gray!30}}l}
\newcolumntype{h}{>{\columncolor{Gray!30}}c}
\begin{tabular}{lggffrrll}
\toprule
\multicolumn{1}{c}{\textbf{method}}   & \multicolumn{4}{c}{\textbf{mean}}   & \multicolumn{4}{c}{\textbf{median}}   \\
\midrule
\rowcolor{GreenYellow}$\cobapr$ & 43.4 & (7.3  & \!:\! & 140.7) &  7 & (5                                    & \!:\!                               & 21)                                   \\
 \mihay                          & 46.1 & (10.9 & \!:\! & 141.3) &  7 & (6                                    & \!:\!                               & 23)                                   \\
 \midrule
 \rowcolor{GreenYellow}$\ncsbmaxrank$ & 30   & (20.3 & \!:\! & 38.3) & 12 & (8                                    & \!:\!                               & 28)                                   \\
 $\ncsblazy$                          & 35.7 & (25.1 & \!:\! & 44.8) & 13 & (9                                    & \!:\!                               & 32)                                   \\
 \bottomrule
\end{tabular}

}
\end{wraptable}
}

In the first part of the experimental evaluation, we measured the effect of the
proposed optimizations from \cref{sec:optimizations} on the size of the generated
state space, i.e., sizes of output automata without any postprocessing.
This use case is motivated by language inclusion and equivalence checking,
where the size of the generated state space directly affects the performance of
the algorithm.
We carried out the
evaluation on \dsltl and \dsautomizer benchmarks (we use \dsboth to denote their union) since most of the automata
there are either IW or SDBAs.

\begin{wraptable}[11]{r}{5.4cm}
\vspace*{3mm}
\caption{Effects of our optimizations for IW and SDBA automata.
  Sizes of output BAs are given as ``\dsboth (\dsltl{} : \dsautomizer)''.
  }
\label{tab:optim}
\vspace{-1mm}
\hspace*{-4mm}
\resizebox{5.9cm}{!}{
\newcolumntype{g}{>{\columncolor{Gray!30}}r}
\newcolumntype{f}{>{\columncolor{Gray!30}}l}
\newcolumntype{h}{>{\columncolor{Gray!30}}c}
\begin{tabular}{lggffrrll}
\toprule
\multicolumn{1}{c}{\textbf{method}}   & \multicolumn{4}{c}{\textbf{mean}}   & \multicolumn{4}{c}{\textbf{median}}   \\
\midrule
\rowcolor{GreenYellow}$\cobapr$ & 43.4 & (7.3  & \!:\! & 140.7) &  7 & (5                                    & \!:\!                               & 21)                                   \\
 \mihay                          & 46.1 & (10.9 & \!:\! & 141.3) &  7 & (6                                    & \!:\!                               & 23)                                   \\
 \midrule
 \rowcolor{GreenYellow}$\ncsbmaxrank$ & 30   & (20.3 & \!:\! & 38.3) & 12 & (8                                    & \!:\!                               & 28)                                   \\
 $\ncsblazy$                          & 35.7 & (25.1 & \!:\! & 44.8) & 13 & (9                                    & \!:\!                               & 32)                                   \\
 \bottomrule
\end{tabular}

}
\end{wraptable}

The first experiment compares the number of states generated by the
original \mihay\ and by the macrostates-pruning optimization
$\cobapr$ from \cref{sec:iw-sim-opt} on inherently weak BAs (948 BAs from
\dsltl{} and 360 BAs from \dsautomizer{} = 1,308~BAs). Note that we omit $\cobasat$ as it is overall worse than $\cobapr$.
The scatter plot is shown in~\cref{fig:res-mihay} and statistics are in the top
part of \cref{tab:optim}.
We can clearly see that the optimization works well, substantially decreasing
both the mean and the median size of the output BAs.

The second experiment compares the size of the state space generated by
\ncsblazy~\cite{ChenHLLTTZ18} and
\ncsbmaxrank from \cref{sec:ncsb-maxrank} on 735~SDBAs (that
are not IW) from \dsltl{} (328~BAs) and \dsautomizer{} (407~BAs).
We omit a comparison with the original NCSB~\cite{BlahoudekHSST16} procedure,
since \ncsblazy behaves overall better~\cite{ChenHLLTTZ18}.
The results are in \cref{fig:res-ncsb} and the bottom part of \cref{tab:optim}.
Again, both the mean and the median are lower for \ncsbmaxrank.
The scatter plot shows that the effect of the optimization is stronger when the
generated state space is larger (for BAs where the output had $\geq$ 150
states, our optimization was never worse).

%

\vspace{-0mm}
\subsection{Comparison with Other Tools}
\vspace{-0mm}


  \begin{table}[t]
  \caption{Statistics for our experiments.
    The table compares the sizes of complement BAs obtained by \ranker and other
    approaches (after postprocessing).
    The \textbf{wins} and \textbf{losses} columns give the number of times when \ranker
    was strictly better and worse.
    The values are given for the three datasets as ``\dsall (\dsrandom{} :
    \dsltl{} : \dsautomizer)''.
    Approaches in \goal are labelled with~\goalmark.
    }
    \vspace{3mm}
  \label{tab:results}
  \hspace{-3mm}
  \resizebox{1.035\linewidth}{!}{%
    \newcolumntype{g}{>{\columncolor{Gray!30}}r}
\newcolumntype{f}{>{\columncolor{Gray!30}}l}
\newcolumntype{h}{>{\columncolor{Gray!30}}c}
\begin{tabular}{lgggggfrrrrrlgggggfrrrrrlgggggf}
\toprule
\multicolumn{1}{c}{\textbf{method}} & \multicolumn{6}{h}{\textbf{mean}}   & \multicolumn{6}{c}{\textbf{median}}   &   \multicolumn{6}{h}{\textbf{wins}} & \multicolumn{6}{c}{\textbf{losses}}   & \multicolumn{6}{h}{\textbf{timeouts}}   \\
\midrule
\rowcolor{GreenYellow}\ranker & 38 & (44 & \!:\! &  9 & \!:\! & 67) & 11 & (18 & \!:\! &  5 & \!:\! & 22) &      &       &       &      &       &      &      &      &       &     &       &      &                                   158 & (53                                 & \!:\!                                 &                                   0 & \!:\!                                 & 105)                                    \\
 $\rankerold$                  & 30 & (38 & \!:\! & 10 & \!:\! & 32) & 12 & (18 & \!:\! &  6 & \!:\! & 22) & 1554 & (356  & \!:\! & 650  & \!:\! & 548) & 264  & (142 & \!:\! & 69  & \!:\! & 53)  &                                   458 & (259                                & \!:\!                                 &                                   7 & \!:\!                                 & 192)                                    \\
 \piterman~\goalmark           & 43 & (56 & \!:\! & 12 & \!:\! & 38) & 14 & (19 & \!:\! &  8 & \!:\! & 24) & 2881 & (1279 & \!:\! & 966  & \!:\! & 636) & 392  & (263 & \!:\! & 68  & \!:\! & 61)  &                                   309 & (12                                 & \!:\!                                 &                                   4 & \!:\!                                 & 293)                                    \\
 \safra~\goalmark              & 49 & (60 & \!:\! & 17 & \!:\! & 56) & 15 & (18 & \!:\! & 10 & \!:\! & 24) & 3109 & (1348 & \!:\! & 1117 & \!:\! & 644) & 274  & (229 & \!:\! & 31  & \!:\! & 14)  &                                   599 & (160                                & \!:\!                                 &                                  30 & \!:\!                                 & 409)                                    \\
 \spot                         & 46 & (57 & \!:\! &  8 & \!:\! & 66) & 11 & (18 & \!:\! &  5 & \!:\! & 18) & 1347 & (935  & \!:\! & 339  & \!:\! & 73)  & 1057 & (327 & \!:\! & 343 & \!:\! & 387) &                                    73 & (13                                 & \!:\!                                 &                                   0 & \!:\!                                 & 60)                                     \\
 \fribourg~\goalmark           & 49 & (68 & \!:\! &  8 & \!:\! & 27) & 11 & (18 & \!:\! &  6 & \!:\! & 19) & 2223 & (1177 & \!:\! & 503  & \!:\! & 543) & 586  & (245 & \!:\! & 207 & \!:\! & 134) &                                   399 & (93                                 & \!:\!                                 &                                   2 & \!:\!                                 & 304)                                    \\
 \ltldstar                     & 44 & (56 & \!:\! & 12 & \!:\! & 47) & 14 & (19 & \!:\! &  7 & \!:\! & 24) & 2794 & (1297 & \!:\! & 924  & \!:\! & 573) & 448  & (283 & \!:\! & 88  & \!:\! & 77)  &                                   288 & (130                                & \!:\!                                 &                                  13 & \!:\!                                 & 145)                                    \\
 \seminator                    & 46 & (58 & \!:\! &  8 & \!:\! & 64) & 11 & (17 & \!:\! &  5 & \!:\! & 21) & 1626 & (1297 & \!:\! & 291  & \!:\! & 38)  & 1113 & (286 & \!:\! & 398 & \!:\! & 429) &                                   419 & (368                                & \!:\!                                 &                                   1 & \!:\!                                 & 50)                                     \\
 \roll                         & 18 & (15 & \!:\! & 11 & \!:\! & 54) &  9 & (8  & \!:\! &  8 & \!:\! & 28) & 6050 & (3824 & \!:\! & 1551 & \!:\! & 675) & 620  & (369 & \!:\! & 125 & \!:\! & 126) &                                  1893 & (1595                               & \!:\!                                 &                                   8 & \!:\!                                 & 290)                                    \\
\bottomrule
\end{tabular}

  }
  \vspace*{-3mm}
  \end{table}

In the second part of the experimental evaluation, we compared \ranker with
other state-of-the-art tools for BA complementation.
We measured how small output BAs we can obtain, therefore, we compared the
number of states after reduction using \autfilt (with the simplification level
\texttt{--high}).
The scatter plots in Fig. \ref{fig:comp-others} compare the numbers of states of
automata generated by $\ranker$, $\rankerold$, and $\spot$. Summarizing
statistics are given in \cref{tab:results}.
The backoff strategy in \ranker was applied in 278 (264:1:13) cases.

\begin{wraptable}[11]{r}{6.3cm}
\vspace{-6mm}
\caption{Run times of the tools [s]
given as ``\dsall (\dsrandom{} : \dsltl{} : \dsautomizer)''}
\label{tab:times}
\vspace{-1mm}
\hspace*{-3mm}
\resizebox{6.6cm}{!}{
\newcolumntype{g}{>{\columncolor{Gray!30}}r}
\newcolumntype{f}{>{\columncolor{Gray!30}}l}
\newcolumntype{h}{>{\columncolor{Gray!30}}c}
\begin{tabular}{lgggggfrrrrrl}
\toprule
\multicolumn{1}{c}{\textbf{method}} & \multicolumn{6}{h}{\textbf{mean}}   & \multicolumn{6}{c}{\textbf{median}}   \\
\midrule
\rowcolor{GreenYellow}\ranker &  3.72 & (4.34  & \!:\! & 0.45 & \!:\! & 7.30)  & 0.05 & (0.10 & \!:\! &                                  0.04 & \!:\!                               & 0.08)                                 \\
 $\rankerold$                  &  4.62 & (5.33  & \!:\! & 0.72 & \!:\! & 9.69)  & 0.07 & (0.19 & \!:\! &                                  0.03 & \!:\!                               & 0.15)                                 \\
 \piterman~\goalmark           &  8.06 & (6.07  & \!:\! & 5.95 & \!:\! & 28.38) & 5.12 & (4.96 & \!:\! &                                  5.08 & \!:\!                               & 8.68)                                 \\
 \safra~\goalmark              & 11.58 & (10.41 & \!:\! & 6.51 & \!:\! & 38.65) & 5.41 & (5.32 & \!:\! &                                  5.26 & \!:\!                               & 9.02)                                 \\
 \spot                         &  0.64 & (0.57  & \!:\! & 0.02 & \!:\! & 2.28)  & 0.02 & (0.02 & \!:\! &                                  0.01 & \!:\!                               & 0.02)                                 \\
 \fribourg~\goalmark           & 13.13 & (14.14 & \!:\! & 6.06 & \!:\! & 23.88) & 5.69 & (6.82 & \!:\! &                                  4.92 & \!:\!                               & 6.57)                                 \\
 \ltldstar                     &  2.1  & (2.25  & \!:\! & 0.34 & \!:\! & 5.15)  & 0.02 & (0.02 & \!:\! &                                  0.01 & \!:\!                               & 0.05)                                 \\
 \seminator                    &  4.16 & (6.33  & \!:\! & 0.03 & \!:\! & 1.88)  & 0.03 & (0.08 & \!:\! &                                  0.01 & \!:\!                               & 0.03)                                 \\
 \roll                         & 23.65 & (29.82 & \!:\! & 3.88 & \!:\! & 49.02) & 3.34 & (6.19 & \!:\! &                                  1.71 & \!:\!                               & 17.14)                                \\
\bottomrule
\end{tabular}

}
\end{wraptable}

First, observe that \ranker significantly outperforms $\rankerold$, especially
in the much lower number of timeouts, which decreased by 65\,\% (moreover, 66~of
the 158 timeouts were due to the timeout of \texttt{autfilt} in postprocessing).
The higher mean of $\ranker$ compared to $\rankerold$ is also caused by less
timeouts).
From \cref{tab:results}, we can also see
that \ranker has the smallest mean and median (except \roll and $\rankerold$,
but they have a~much higher number of timeouts).
\ranker has also the second lowest number of timeouts (\spot has the lowest).
If we look at the number of \textbf{wins} and \textbf{loses}, we can see that
\ranker in majority of cases produces a strictly smaller automaton compared to
other tools.
In \cref{tab:times}, see that the run time of \ranker is comparable to the run
times of other tools (much better than \goal and \roll, comparable with
\seminator, and a~bit worse than \spot and \ltldstar).



\paragraph{Acknowledgements.}
We thank the anonymous reviewers for their useful remarks that helped us improve
the quality of the paper, the artifact evaluation committee for their
thorough testing of the artifact,
and Alexandre Duret-Lutz for useful feedback on an earlier version of the paper.
This work was supported by the Czech Ministry of Education, Youth and Sports project LL1908 of the ERC.CZ programme,
the Czech Science Foundation project 20-07487S,
and the FIT BUT internal project FIT-S-20-6427.

\newpage

\bibliographystyle{splncs}
\bibliography{literature}

\newpage
\appendix

\section{Proofs of Section~\ref{sec:iw-sim-opt}}

\theFairSim*

\begin{proof}
	Assume that $p \dirsimbyw q$ for some $p,q\in Q$.
  To show that $p \fairsimbyc q$, we need to show that for every accepting
  trace~$\pi$ of~$\autc$ from~$p$ (i.e., since~$\autc$ is a~co-BA, a~trace with finitely
  many accepting states from $Q \setminus \accstates$), there is an accepting
  trace~$\pi'$ of~$\autc$ from~$q$ over the same word~$\word$.
  Since~$\pi$ is accepting in~$\autc$, there is some $n\in\omega$ s.t.\ for all
  $\ell \geq n$ it holds that $\pi_\ell \notin Q\setminus \accstates$, i.e.,
  $\pi_\ell \in \accstates$.
  Then we can construct~$\pi'$ as a~trace of~$\autc$ over~$\word$ that
  direct-simulates (w.r.t.\ $\dirsimbyw$) the trace~$\pi$.
  Because of the properties of direct simulation, it holds that for all $\ell
  \geq n$ we have that $\pi'_\ell \in \accstates$, i.e., $\pi'_\ell \notin Q
  \setminus \accstates$.
  Hence, $\pi'$ is
	accepting in the co-BA~$\autc$.
  \qed

\end{proof}

The rest of this section is devoted to the proof of Theorem~\ref{the:coba-corr}.
For that reason we introduce definitions and notions used further. In the
following we fix a co-BA $\autc = (Q, \delta, I, \accstates, \emptyset)$. We
use $p\leadsto q$ to denote that $q$ is reachable from $p$. Let
$\alpha\in\Sigma^\omega$ be a word. Let $\Pi, \Pi'$ be sets of traces over
$\alpha$. We say that $\Pi$ and $\Pi'$ are \emph{acc-equivalent}, denoted as
$\Pi \sim \Pi'$ if $\exists\pi\in\Pi: \pi$ is accepting in $\autc$ iff
$\exists\pi'\in\Pi': \pi'$ is accepting in $\autc$.
Let $\rho = S_1S_2\dots$ be a sequence of sets of states and $\alpha$ be a word.
We define $\Pi_\rho$ to be a set of traces over $\alpha$ matching the sets of
states. Formally, $\Pi_\rho = \{ \pi \mid \pi \text{ over } \alpha, \pi_i \in
S_i \text{ for each } i \}$. We also define $\Pi_\rho^\cup =
\bigcup_{i\in\omega}\Pi_{\rho_{i:\omega}}$. Further, for a set of state $B$ we
use $\rho_\alpha^B$ to denote the sequence $S_1S_2\dots$ s.t. $S_1 = B$,
$S_{i+1} = \delta(S_i, \alpha_i)$ for each $i\in\omega$. We use $\rho_\alpha$ to
denote $\rho_\alpha^I$. Moreover, for a given mapping $\theta: 2^Q \to 2^Q$ and
a sequence of sets of states $\rho$ we define $\theta(\rho) = \theta(\rho_1)
\theta(\rho_2)\dots$. A trace $\pi$ is \emph{eventually fair-simulated} by
$\pi'$ if there is some $i \in\omega$ s.t. $\pi_{i:\omega} \fairsimbyc
\pi'_{i:\omega}$.

\begin{lemma}\label{lem:trace-union}
	Let $\alpha$ be a word, $\Pi_{\rho_\alpha}\sim \Pi_{\rho_\alpha'}$, and
	$\Pi_{\rho_\alpha}\subseteq \Pi_{\rho_\alpha'}$. Then, $\Pi_{\rho_\alpha}\sim
	\Pi_{\rho_\alpha'}^\cup$.
\end{lemma}
\begin{proof}
	Assume that $\Pi_{\rho_\alpha}\sim \Pi_{\rho_\alpha'}$, and
	$\Pi_{\rho_\alpha}\subseteq \Pi_{\rho_\alpha'}$. Since $\Pi_{\rho_\alpha'}
	\subseteq \Pi_{\rho_\alpha'}^\cup$, if there is an accepting trace in
	$\Pi_{\rho_\alpha}$, there is (the same) accepting trace in
	$\Pi_{\rho_\alpha'}^\cup$. If there is no accepting trace in
	$\Pi_{\rho_\alpha}$, it means that all traces contain infinitely many accepting
	states. Hence, every infinite suffix is also an accepting trace and therefore,
	$\Pi_{\rho_\alpha'}^\cup$ contain all traces that are not accepting (with
	infinitely many accepting states).
	\qed
\end{proof}

\begin{lemma}\label{lem:pr-trace}
	Let $\alpha$ be a word. Then, $\Pi_{\rho_\alpha}\sim \Pi_{\pr(\rho_\alpha)}$.
\end{lemma}
\begin{proof}
	First observe that $\Pi_{\pr(\rho_\alpha)} \subseteq \Pi_{\rho_\alpha}$.
	Therefore, it suffices to show that if there is an accepting trace $\pi \in
	\Pi_{\rho_\alpha}$, there is also an accepting trace $\pi' \in
	\Pi_{\pr(\rho_\alpha)}$. We assume that and we show that there is $\pi'
	\in \Pi_{\pr(\rho_\alpha)}$ s.t. $\pi$ is eventually fair-simulated by
	$\pi'$. If $\pi' = \pi$ we are done. Now, assume that it is not the case and
	that there is a maximum set of traces $P = \{ \pi^1, \pi^2, \dots \}\subseteq
	\Pi_{\rho_\alpha}$ with indices $\ell_1 < \ell_2 < \dots$ s.t. $p_i =
	\pi^i_{\ell_i} \sqsubseteq \pi^{i+1}_{\ell_i} = p_i'$ for each $i$, and
	moreover $\pi_1 = \pi$. We show that $P$ is in fact finite by showing that
	$p'_i \neq p'_j$ for each $i\neq j$. Assume that $p_j' = p'_i$ for some $i<j$.
	But then we have $p_i' \leadsto p_j\sqsubseteq p_j' = p_i'$ meaning that $p_i'
	\leadsto p_j'$ (from the definition of $\sqsubseteq$). From the definition of
	$\sqsubseteq$ we also have $p_j$ is not reachable from $p_j' = p_i'$, which is
	a contradiction. Since the set $P = \{ \pi_1, \dots, \pi_n \}$ is maximum and
	finite, we have $\pi_n \in \Pi_{\pr(\rho_\alpha)}$. Moreover, $\pi' =
	\pi_n$ eventually fair-simulates $\pi$ (given by the step-wise property of
	fair simulation), which concludes the proof.
	\qed
\end{proof}

\begin{lemma}\label{lem:sat-trace}
	Let $\alpha$ be a word. Then, $\Pi_{\rho_\alpha}\sim \Pi^\cup_{\sat(\rho_\alpha)}$.
\end{lemma}
\begin{proof}
	First observe that $\Pi_{\rho_\alpha} \subseteq \Pi^\cup_{\sat(\rho_\alpha)}$.
	Therefore, it suffices to show that if there is an accepting trace $\pi \in
	\Pi^\cup_{\sat(\rho_\alpha)}$, there is also an accepting trace $\pi' \in
	\Pi_{\rho_\alpha}$. We fix $\rho = \rho_\alpha$. Consider some accepting trace
	$\pi \in \Pi^\cup_{\sat(\rho_\alpha)}$. If $\pi \in \Pi_{\rho_\alpha}$, we are
	done. If not, there is some position $\ell$ s.t. $\pi \in
	\Pi_{\rho_{\ell:\omega}}$ and $\pi_1 \fairsimby q$ where $q\in\rho_\ell$.
	Therefore, there is some trace $\pi' \in\rho$ s.t. $\pi'_\ell = q$. Moreover,
	$\pi$ is accepting, hence there is a trace $\pi''$ leading from $q$, which is
	accepting as well. Hence, $\pi'_{1:\ell}.\pi'' \in \rho$ and moreover this
	trace is accepting.
	\qed
\end{proof}

\begin{lemma}\label{lem:adjust-corr}
	Let $\theta$ be an adjusting function. If $\Pi_{\rho_\alpha} \sim \Pi_{\theta(\rho_\alpha)}^\cup$ for each $\alpha
	\in\Sigma^\omega$ then $\langof{\cobapar(\autc)} = \Sigma^\omega \setminus
	\langof{\autc}$.
\end{lemma}
\begin{proof}
	\emph{(Sketch)}
	Consider a word $\alpha\in\langof{\autc}$. Hence, there is an accepting trace
	$\pi\in\Pi_{\rho_\alpha}$ and also an accepting trace $\pi' \in
	\Pi_{\theta(\rho_\alpha)_{k:\omega}}$ for some $k \in\omega$. Since, $\pi'$
	emerges eventually in the $B$ set, $\alpha$ is not accepted by
	$\langof{\cobapar(\autc)}$.

	Conversely, assume that $\alpha\not\in\langof{\autc}$. Then, all traces in
	$\Pi_{\rho_\alpha}$ as well in $\Pi_{\theta(\rho_\alpha)}^\cup$ contain
	infinitely many accepting states. Hence, we flush $B$-set infinitely many
	times. yielding $\alpha \in \langof{\cobapar(\autc)}$.
	\qed
\end{proof}

\theCoBaCorr*

\begin{proof}
	Theorem~\ref{the:coba-corr} for $\cobapr(\autc)$ we get directly from the fact
	that $\Pi_{\pr(\rho_\alpha)} \subseteq \Pi_{\rho_\alpha}$ for any word $\alpha$,
	and from Lemmas~\ref{lem:trace-union}, \ref{lem:pr-trace}, and
	\ref{lem:adjust-corr}. Correctness of for $\cobasat(\autc)$ we obtain from
	Lemmas~\ref{lem:sat-trace} and~\ref{lem:adjust-corr}.
\end{proof}

\section{Proofs of Section~\ref{sec:ncsb-maxrank}}

\begin{lemma}\label{lem:sd-b}
	Let $B\subseteq \stdet$ be a set of deterministic states and let $\alpha$ be a
	word. If $\alpha \notin \langof\aut$, then $\exists k: \forall \ell \geq k:
	(\rho^B_\alpha)_\ell \cap \accstates = \emptyset \wedge
	\acctrans((\rho^B_\alpha)_\ell, \alpha_\ell) = \emptyset$.
\end{lemma}
\begin{proof}
	Assume that $\alpha \notin \langof\aut$. Since $B$ is a set of states of the
	deterministic part, we have $|\Pi_{\rho^B_\alpha}| \leq |B|$. Moreover, for
	each trace $\pi\in \Pi_{\rho^B_\alpha}$ there is some $k_\pi$ s.t. $\pi_\ell
	\notin \accstates$ and $\pi_{\ell+1} \notin \acctrans(\pi_\ell, \alpha_\ell)$ for each
	$\ell \geq k_\pi$. Taking $k = \max\{ k_\pi \mid \pi\in \Pi_{\rho^B_\alpha}
	\}$ we fulfill the condition of the lemma.
	\qed
\end{proof}

\theMaxrankCorr*

\begin{proof}
	First, we prove that if $\alpha \in \langof{\aut}$, then $\alpha\notin
	\langof{\ncsbmaxrank(\aut)}$. In that case, there is a run $\rho$ on $\alpha$
	in $\aut$. Moreover, $\rho_\ell \in \stdet$ for some $\ell \in \omega$.
	Therefore, for every run $R = (N_1,C_1,S_1,B_1)\dots$ on $\alpha$ in
	$\ncsbmaxrank(\aut)$, we have that either $\rho_\ell \in S_\ell$ or $\rho_\ell
	\in C_\ell$. Now assume the first case, $\rho_\ell\in S_\ell$. At some point,
	we reach an accepting state in $\rho$ ($\rho_k \in \accstates$, $k\geq \ell$) or an
	accepting transition ($\rho_{k+1} \in\acctrans(\rho_k, \alpha_k)$, $k\geq
	\ell$). Hence $\rho_k \in S_k$ means that $R$ is a finite trace of at most
	$k-1$ elements (recall that $S\cap \accstates = \emptyset$ and the
	condition~$\acctrans(S,a) = \emptyset$). Now we turn to the second case, $\rho_\ell \in
	C_\ell$. In that case, either $\rho_l \in C_l$ and $\rho_l \in B_l$ for each
	$l \geq l_0 \geq \ell$, or we apply $\gamma_2$ and move $\rho$ to $S$, i.e.,
	$\rho_m \in S_m$ for some $m \geq \ell$. In the first case, $B$ is not empty
	anymore, hence $R$ is not accepting. In the later, we get the case examined
	before yielding to a finiteness of $R$. Hence, $\alpha\notin
	\langof{\ncsbmaxrank(\aut)}$.

	Now, we prove that if $\alpha \notin \langof{\aut}$, then $\alpha\in
	\langof{\ncsbmaxrank(\aut)}$. We construct an accepting run $R$ on $\alpha$ in
	$\aut$. Let $R_0 = (N_1, C_1, S_1, B_1) = (\stnondet\cap I, \stdet\cap I, \emptyset,
	\stdet \cap I)$. From Lemma~\ref{lem:sd-b} we have that there is a $k_1$ s.t.
	$\forall \ell \geq k_1: (\rho^{B_1}_\alpha)_\ell \cap \accstates = \emptyset \wedge
	\acctrans((\rho^{B_1}_\alpha)_\ell, \alpha_\ell) = \emptyset$. We set $R_{i+1}
	= \gamma_1(R_{i})$ for $1 \leq i < k$. Further, we set $R_{k+1} =
	\gamma_2(R_k)$. Then, in the same sense we use Lemma~\ref{lem:sd-b} (on
	$\alpha_{k_1:\omega}$) to obtain a position $k_2$ giving us the point where
	$\gamma_2$ is applied. Such constructed run $R$ is infinite, because
	Lemma~\ref{lem:sd-b} ensures that we cannot reach an accepting state or an
	accepting transition from $S$ on $\alpha$. It remains to show that $R$ is
	accepting. From the construction, we have that $\gamma_2$ was used infinitely
	many times (and each successor of $\gamma_2$ is an accepting state). Hence the run contains infinitely many accepting states.
	%
	%
	\qed
\end{proof}

\end{document}